\documentclass[11pt,a4paper]{article}
\usepackage[margin=1in]{geometry}
\usepackage{amsmath,amssymb,amsthm}
\usepackage{hyperref}
\usepackage{booktabs}
\usepackage{microtype}
\usepackage{natbib}
\usepackage{enumitem}
\usepackage{mathtools}
\usepackage{xcolor}
\usepackage{array}

\newtheorem{theorem}{Theorem}[section]
\newtheorem{lemma}[theorem]{Lemma}
\newtheorem{corollary}[theorem]{Corollary}

\newtheorem{definition}[theorem]{Definition}
\newtheorem{remark}[theorem]{Remark}

\newcommand{\cl}{\operatorname{cl}}
\newcommand{\Safe}{\mathcal{R}}
\newcommand{\B}{\mathcal{B}}
\newcommand{\Wp}{\mathsf{Why}}
\newcommand{\Emg}{\operatorname{Emg}}
\newcommand{\NMF}{\operatorname{NMF}}
\newcommand{\Datalog}{\text{Datalog}_{\mathsf{prop}}}

\newcommand{\citepos}[1]{\citeauthor{#1}'s \citeyearpar{#1}}

\hypersetup{
  pdftitle={Capability Safety as Datalog: A Foundational Equivalence},
  pdfauthor={Cosimo Spera},
  pdfkeywords={Datalog; capability hypergraphs; AI safety; incremental view maintenance; provenance semirings}
}

\title{\textbf{Capability Safety as Datalog:\\
A Foundational Equivalence}}

\author{Cosimo Spera\thanks{Corresponding author.
  Minerva CQ, 114 Lester Ln, Los Gatos, CA~95032.
  \texttt{cosimo@minervacq.com}}}

\date{March 2026}

\begin{document}
\maketitle

\begin{abstract}
\noindent
We prove that capability safety admits an exact representation as propositional
Datalog evaluation ($\Datalog$: the monadic, ground, function-free fragment of
first-order logic), enabling the transfer of algorithmic and structural results
unavailable in the native formulation.  This addresses two structural limitations of
the capability hypergraph framework of \citet{Spera2026}: the absence of efficient
incremental maintenance, and the absence of a decision procedure for audit surface
containment.  The equivalence is tight: capability hypergraphs correspond to exactly
this fragment, no more.

\smallskip\noindent\textbf{Algorithmic consequences.}
Under this identification, the safe goal discovery map $G_F(A)$ is a stratified
Datalog view.  This structural fact yields the \textbf{Locality Gap Theorem}, which
establishes the first structural separation between global recomputation and local
capability updates: DRed maintenance costs $O(|\Delta|\cdot(n+mk))$ per update versus
$O(|V|\cdot(n+mk))$ for recomputation; an explicit hard family witnesses an $\Omega(n)$
asymptotic gap; and an AND-inspection lower bound---proved via indistinguishable instance
pairs in the oracle model---shows any correct algorithm must probe all $k+1$ atoms in
$\Phi(u)=S_u\cup\{v_u\}$ to verify rule activation.
Together, these yield the first provable separation between global and local safety
reasoning in agentic systems.

\smallskip\noindent\textbf{Structural and semantic consequences.}
Structurally: audit surface containment $G_F(A)\subseteq G_F(A')$ is decidable in
polynomial time via $\Datalog$ query containment, giving the first decision procedure
for this problem; non-compositionality of safety is the non-modularity of Datalog
derivations, providing a structural explanation rather than a case analysis.
Semantically: derivation certificates are why-provenance witnesses
\citep{Green2007}, with a commutative semiring algebra enabling compression,
composition, and uniform validation.  Each open problem of \citet{Spera2026} maps
to a known open problem in Datalog theory, enabling direct transfer of thirty years
of partial results.

\smallskip\noindent\textbf{Scope.}
The syntactic connection between AND-hyperedges and Horn clauses is classical;
the contribution is to make it a tight formal equivalence and to derive consequences
from it that are new.  This gives the capability safety research agenda direct access
to thirty years of Datalog results, without requiring new algorithms or new
complexity analysis for problems already solved in database theory.

\end{abstract}

\tableofcontents
\newpage

%%==========================================================================
\section{Introduction}
%%==========================================================================

\subsection{A New Computational Lens for Agentic Safety}

The safety of multi-agent AI systems is ultimately a computational problem: given a
set of capabilities agents currently hold and a set of forbidden states, determine
what the agents can collectively reach, how efficiently that determination can be
maintained as the system evolves, and what certificates of safety can be produced for
regulators and auditors.  \citet{Spera2026} formalised this as a capability
hypergraph problem and proved non-compositionality of safety, but the computational
tools available within the hypergraph framework are limited: closure is computed by a
fixed-point iteration, and every structural change requires recomputation from scratch.

\medskip
\noindent\textbf{The resolution.}
We resolve both limitations by proving that capability safety is exactly evaluation
in propositional Datalog.  We prove this as a formal equivalence, with
explicit polynomial-time encodings in both directions.  The syntactic connection
between AND-hyperedges and Horn clauses has been noted before; what is new is making
it exact and using it to transfer algorithmic results that have no counterpart in
the hypergraph framework.

\medskip
\noindent\textbf{Why this is not trivial.}
Syntactic correspondences between hypergraphs and Horn clauses are known, but they do
not imply semantic or algorithmic equivalence.  In particular, they do not preserve
safety closure, minimal unsafe sets, or support incremental maintenance.  The
contribution of this work is to establish an exact equivalence \emph{with these
properties}---showing that the capability hypergraph closure, the structure of $\B(F)$,
and the safe audit surface $G_F(A)$ all have precise $\Datalog$ counterparts---and to
use this to derive new structural and algorithmic results, chiefly the Locality Gap
Theorem, that are not accessible in the hypergraph framework alone.

\medskip
\noindent\textbf{Algorithmic consequences.}
The primary payoff is the \textbf{Locality Gap Theorem}
(Theorem~\ref{thm:locality-gap}).  In \citet{Spera2026}, every hyperedge update
requires recomputing the safe audit surface $G_F(A)$ from scratch at cost
$O(|V|\cdot(n+mk))$.  The $\Datalog$ identity shows $G_F(A)$ is a view, making DRed
incremental maintenance available at cost $O(|\Delta|\cdot(n+mk))$ per update.  An
explicit hard family witnesses an $\Omega(n)$ asymptotic gap.  An AND-inspection
lower bound---proved via indistinguishable instance pairs under an oracle model,
invoking Yao's minimax principle \citep{Yao1977}---shows any correct algorithm must
probe all $k+1$ atoms in $\Phi(u)=S_u\cup\{v_u\}$ to verify rule activation.  The
$\Omega(k)$ bound reflects the AND-condition structure directly: it is the same
conjunctive precondition structure that makes safety non-compositional.

\medskip
\noindent\textbf{Structural and semantic consequences.}
The identity gives two further results unavailable in the hypergraph framework.
Deciding whether one audit surface is contained in another, $G_F(A)\subseteq
G_F(A')$, reduces to $\Datalog$ query containment, giving the first polynomial-time
decision procedure for this problem.  Derivation certificates are why-provenance
witnesses \citep{Green2007}, with a commutative semiring structure enabling
compression, composition, and uniform validation within a single algebraic framework.
Each open problem of \citet{Spera2026} admits a reduction to a known open problem in Datalog
learning theory, probabilistic Datalog, or view update, enabling direct transfer of
thirty years of partial results.

\smallskip
\noindent\emph{This result shows that capability safety does not require a new
computational theory, but can be analysed within an existing, well-understood
logical framework---with immediate access to its algorithms, complexity results,
and open problems.}

\subsection{What is Datalog?}

Datalog is the fragment of first-order logic consisting of rules of the form
$B_1(\mathbf{x}_1) \wedge \cdots \wedge B_k(\mathbf{x}_k) \Rightarrow H(\mathbf{y})$,
where all variables are universally quantified, no function symbols appear, and $H$ is
a single relational atom.  Given a finite database $D$, the semantics of a Datalog program
$\Pi$ over $D$ is the least fixed point of the immediate consequence operator
$T_\Pi(I) = D \cup \{H\sigma : (B_1 \wedge \cdots \wedge B_k \Rightarrow H) \in \Pi,\;
B_i\sigma \in I\ \forall i\}$,
computable in polynomial time in $|D|$ \citep{AbiteboulHullVianu1995,Ceri1989}.

\subsection{The Propositional Fragment}

When all predicates are unary and all terms are constants, Datalog reduces to propositional
Horn clause logic.  Rules have the form $p_1 \wedge \cdots \wedge p_k \Rightarrow q$ where
$p_i, q$ are propositional atoms.  We denote this fragment $\Datalog$.

The capability hypergraph framework lives precisely in this fragment.  Each capability
$v \in V$ is a propositional atom $\mathsf{has}(v)$; each hyperedge $(S,\{v\})$ is a
propositional Datalog rule; and the closure $\cl(A)$ is the least model of the corresponding
Datalog program over $\{\mathsf{has}(a) : a \in A\}$.

\subsection{Contributions}

\begin{enumerate}[leftmargin=*]
\item \textbf{Locality Gap Theorem} (Theorems~\ref{thm:audit-view}
  and~\ref{thm:locality-gap}): the safe audit surface $G_F(A)$ is a stratified
  $\Datalog$ view, yielding the first provable separation between global and local
  safety reasoning in agentic systems.  DRed maintenance costs
  $O(|\Delta|\cdot(n+mk))$ versus $O(|V|\cdot(n+mk))$ for recomputation; an explicit
  hard family witnesses an $\Omega(n)$ asymptotic gap; an AND-inspection lower bound
  (proved via an oracle model and Yao's minimax principle) shows any correct algorithm
  must probe all $k+1$ atoms in $\Phi(u)=S_u\cup\{v_u\}$ to verify rule activation.
\item \textbf{The Encoding Theorem} (Theorems~\ref{thm:encoding-cap-to-dl}
  and~\ref{thm:encoding-dl-to-cap}): a tight, two-directional equivalence between
  capability hypergraphs and $\Datalog$ programs, with explicit polynomial-time
  isomorphisms in both directions preserving closure, safety, and $\B(F)$.  The scope
  is propositional throughout; a parity-query separation
  (Theorem~\ref{thm:expressivity}) proves the fragment cannot be enlarged.
\item \textbf{Containment Decidability}: audit surface containment
  $G_F(A)\subseteq G_F(A')$ reduces to $\Datalog$ query containment, giving the first
  polynomial-time decision procedure for this problem.
\item \textbf{Provenance Algebra} (Theorem~\ref{thm:provenance}): derivation
  certificates are why-provenance witnesses \citep{Green2007} with a commutative
  semiring structure; $\B(F)$ equals the minimal witness antichain
  (Theorem~\ref{thm:witness}).
\item \textbf{Open Problem Transfer} (Section~\ref{sec:open}): each open problem of
  \citet{Spera2026} admits a reduction to a known open problem in Datalog theory, importing thirty
  years of partial results.
\item \textbf{Empirical Illustration} (Section~\ref{sec:empirical}): the 900-trajectory
  corpus of \citet{Spera2026} is re-expressed in Datalog vocabulary as illustration;
  detailed derivation trees appear in Appendix~\ref{app:empirical}.
\end{enumerate}

%%==========================================================================
\section{Related Work}
\label{sec:related}
%%==========================================================================

This paper sits at the intersection of three bodies of literature: Datalog theory and
finite model theory, AI safety formalisms, and logic-based learning theory.  We position
the contribution relative to each.

\paragraph{Datalog and descriptive complexity.}
The foundational result underlying this paper is \citepos{Immerman1986} theorem that
$\Datalog$ captures polynomial-time query evaluation over ordered structures, and the
subsequent work of \citet{Fagin1974} on the relationship between Datalog and fixed-point
logics.  Our Corollary~\ref{cor:complexity} re-derives the complexity results of
\citet{Spera2026} from this landscape, strengthening them by situating them in the
descriptive complexity characterisation of $\mathsf{P}$ rather than ad-hoc circuit
reductions.

The $\mathsf{coNP}$-completeness of computing $\B(F)$ (Corollary~\ref{cor:complexity}
Part~3) follows the classical result of \citet{EiterGottlob1995} on minimal transversal
enumeration, which is one of the central open problems in database theory.  The connection
we establish --- that $\B(F)$ is exactly the minimal-witness antichain for a monotone
Datalog query --- places this open problem in a well-studied context with a 30-year
literature of partial results.

\paragraph{Datalog provenance and semirings.}
The provenance semiring framework of \citet{Green2007} is the direct foundation for our
Provenance Theorem (Theorem~\ref{thm:provenance}).  \citet{Green2007} showed that Datalog
derivations can be annotated with elements of any commutative semiring, yielding a
uniform algebraic account of why-provenance, lineage, and uncertainty propagation.  Our
contribution is to show that the derivation certificates of \citet{Spera2026} are exactly
the why-provenance witnesses under the encoding, giving them a canonical algebraic
structure that was not previously recognised.  Prior work on Datalog provenance has
focused on relational queries \citep{AbiteboulHullVianu1995}; this paper extends the
connection to the AI safety setting.

\paragraph{Logic-based learning and Horn clause learning.}
\citet{Cohen1995} studied the learnability of Horn programs from positive and negative
examples, establishing PAC-learning bounds for restricted Horn clause classes.  His
work is the direct precursor to the open question we raise in
Section~\ref{sec:open}: whether the VC-dimension bound of $O(n^{2k})$ from
\citet{Spera2026} can be tightened for structured capability hypergraph families.
\citet{Dalmau2002} studied Datalog learning in the context of constraint satisfaction
and bounded treewidth, providing Rademacher-complexity bounds that are substantially
tighter than the Sauer--Shelah bound for programs with restricted dependency structure.
The capability hypergraph setting --- typically sparse, low fan-in, and tree-structured
--- is precisely the regime where these tighter bounds apply.

\paragraph{AI safety formalisms.}
The broader AI safety literature has proposed several formal frameworks for reasoning
about agent behaviour.  \citet{Leike2017} defines safety in terms of reward functions
over environment histories; this is orthogonal to the capability-composition setting,
which is concerned with structural reachability rather than optimisation objectives.
Contract-based design \citep{Benveniste2018} and assume-guarantee reasoning
\citep{Jones1983} verify that a \emph{fixed} composition satisfies a pre-specified
property; the capability hypergraph framework characterises the set of \emph{all}
properties a dynamically growing capability set can ever reach, which is a strictly
more demanding problem.  The closest formalisms in the planning literature --- Petri
net reachability \citep{Murata1989} and AND/OR planning \citep{Erol1994} --- encode
conjunctive preconditions, but require search rather than a single fixed-point
computation, and provide neither the non-compositionality theorem nor a polynomial-time
certifiable audit surface.  The present paper's contribution is not to replace these
formalisms but to precisely characterise the capability hypergraph framework's position
in the expressivity landscape by identifying it with $\Datalog$.

\paragraph{View update and incremental maintenance.}
The view update problem for Datalog --- given a Datalog view and a desired change to the
view, compute the minimal change to the base data --- is a classical open problem
\citep{AbiteboulHullVianu1995}.  Section~\ref{sec:open} shows that the adversarial
robustness problem of \citet{Spera2026} (\textsc{MinUnsafeAdd}) is an instance of this
problem, placing its $\mathsf{NP}$-hardness for $b\geq 2$ in the context of known
hardness results for general view update.  The DRed algorithm \citep{AbiteboulHullVianu1995}
and the semi-na\"ive evaluation strategy provide the incremental maintenance machinery
that Theorem~\ref{thm:audit-view} applies to the safe audit surface.

%%==========================================================================
\section{Formal Background}
%%==========================================================================

\subsection{Capability Hypergraphs}

We summarise the relevant definitions from \citet{Spera2026} in full, so that the present
paper is self-contained for readers with a Datalog background who may not have read the
source paper.

\begin{definition}[Capability Hypergraph]
A \emph{capability hypergraph} is a pair $H=(V,\mathcal{F})$ where $V$ is a finite set
of \emph{capability nodes} and $\mathcal{F}$ is a finite set of \emph{hyperarcs}, each
of the form $e=(S,T)$ with $S,T\subseteq V$ and $S\cap T=\emptyset$.  The set $S$ is
the \emph{tail} (joint preconditions) and $T$ is the \emph{head} (simultaneous effects).
The hyperarc fires when all elements of $S$ are simultaneously present, producing all
elements of $T$.  Directed graphs are the special case $|S|=|T|=1$.
\end{definition}

\begin{definition}[Closure Operator]
Let $H=(V,\mathcal{F})$ and $A\subseteq V$.  The \emph{closure} $\cl_H(A)$ is the
smallest set $C\subseteq V$ satisfying: (i)~$A\subseteq C$ (extensivity); and
(ii)~$\forall(S,T)\in\mathcal{F}:\;S\subseteq C\Rightarrow T\subseteq C$ (closed under
firing).  It is computed by the fixed-point iteration $C_0=A$,
$C_{i+1}=C_i\cup\{T\mid(S,T)\in\mathcal{F},\,S\subseteq C_i\}$, which terminates in
at most $|V|$ steps.  The operator satisfies extensivity, monotonicity
($A\subseteq B\Rightarrow\cl(A)\subseteq\cl(B)$), and idempotence.
The worklist implementation runs in $O(n+mk)$ where $n=|V|$, $m=|\mathcal{F}|$, and
$k=\max_{e\in\mathcal{F}}|S(e)|$.
\end{definition}

\begin{definition}[Safe Region and Minimal Unsafe Antichain]
Let $F\subseteq V$ be a \emph{forbidden set}.  A configuration $A\subseteq V$ is
\emph{$F$-safe} if $\cl_H(A)\cap F=\emptyset$.  The \emph{safe region} is
$\Safe(F)=\{A\subseteq V:\cl_H(A)\cap F=\emptyset\}$.

The safe region is a lower set (downward-closed) in $(2^V,\subseteq)$: if $A\in\Safe(F)$
and $B\subseteq A$ then $B\in\Safe(F)$ \citep[Theorem~9.4]{Spera2026}.  Consequently
its complement $\overline{\Safe(F)}$ is an upper set, whose \emph{antichain of minimal
elements} is the \emph{minimal unsafe antichain}:
\[
  \B(F) = \{A\subseteq V : A\notin\Safe(F),\;\forall a\in A,\;A\setminus\{a\}\in\Safe(F)\}.
\]
By Dickson's lemma, $\B(F)$ is finite.  Its key computational property
\citep[Theorem~9.5]{Spera2026}: deciding $B\in\B(F)$ is $\mathsf{coNP}$-complete.
\end{definition}

\begin{definition}[Emergent Capabilities and Near-Miss Frontier]
\label{def:emg-nmf}
Fix $H=(V,\mathcal{F})$, $A\subseteq V$, and let $C=\cl_H(A)$.  Let $\cl_1(A)$ denote
the closure under the sub-hypergraph of singleton-tail hyperarcs only.  Then:
\begin{itemize}
\item \emph{Emergent capabilities:} $\Emg(A)=\{v\in C\setminus A : v\notin\cl_1(A)\}$.
  These are capabilities reachable from $A$ via conjunctive (AND) hyperarcs but not via
  any chain of singleton-tail arcs alone.
\item \emph{Near-miss frontier:} $\NMF_F(A)=\{\mu(e):e\in\partial(A),\;\mu(e)\notin F,\;
  \cl(A\cup\mu(e))\cap F=\emptyset\}$, where the \emph{boundary} $\partial(A)$ is the
  set of hyperarcs $e=(S,\{v\})$ with $S\not\subseteq C$ and $|S\setminus C|=1$, and
  $\mu(e)=S\setminus C$ is the single missing precondition.
\end{itemize}
\end{definition}

\begin{definition}[Safe Audit Surface]
The \emph{safe goal discovery map} of \citet[Theorem~10.1]{Spera2026} is:
\[
  G_F(A) = \Bigl(\Emg(A)\setminus F,\;\NMF_F(A),\;
    \operatorname{top}\text{-}k_{v\in V\setminus(\cl(A)\cup F)}\gamma_F(v,A)\Bigr),
\]
where the marginal gain $\gamma_F(v,A)=|\cl(A\cup\{v\})\setminus(\cl(A)\cup F)|$.
It is polynomial-time computable in $O(|V|\cdot(n+mk))$ and provides derivation
certificates for every element of $\Emg(A)\setminus F$.
\end{definition}

\subsection{Propositional Datalog}

\begin{definition}[Propositional Datalog Program]
A \emph{propositional Datalog program} is a pair $\Pi=(R,D_0)$ where $R$ is a finite set
of propositional Horn rules $p_1\wedge\cdots\wedge p_k \Rightarrow q$ ($k\geq 0$), and
$D_0$ is a finite set of ground facts (the extensional database, EDB).  The \emph{least
model} of $\Pi$ is $\Pi(D_0)=T_\Pi\!\uparrow\!\omega$, computed by iterating the
immediate consequence operator to its least fixed point.
\end{definition}

\begin{definition}[Datalog Query and Minimal Witness]
A \emph{propositional Datalog query} is a pair $(\Pi,q)$ where $q$ is a distinguished
atom.  The query evaluates to \emph{true} over database $D$ iff $q\in\Pi(D)$.
A \emph{minimal witness} for $(\Pi,q)$ over a database family $\{D_A : A\subseteq V\}$
is a set $W\subseteq V$ such that $q\in\Pi(D_W)$ and $q\notin\Pi(D_{W\setminus\{w\}})$
for every $w\in W$.
\end{definition}

\subsection{Provenance Semirings}

\citet{Green2007} showed that Datalog derivations can be annotated with elements of a
commutative semiring $(K,+,\times,0,1)$, yielding a \emph{provenance polynomial} recording
which combinations of base facts contributed to each derived fact.  The
\emph{why-provenance semiring} $(2^{2^V},\cup,\bowtie,\emptyset,\{\emptyset\})$ records
for each derived atom the set of minimal EDB subsets sufficient to derive it, where
$\bowtie$ denotes pairwise set union.

%%==========================================================================
\section{The Encoding Theorem}
\label{sec:encoding}
%%==========================================================================

\noindent\emph{Intuition.}
Capability hypergraphs encode conjunction through hyperedges: a hyperedge
$(S,\{v\})$ fires precisely when \emph{all} preconditions in $S$ are simultaneously
present.  Datalog rules encode the same structure through Horn clauses:
$\bigwedge_{s\in S}\mathsf{has}(s)\Rightarrow\mathsf{has}(v)$.  The key difficulty
is not representational but \emph{semantic}: ensuring that safety closure, minimal
witnesses, and incremental update behaviour are preserved exactly under the mapping.
The theorems in this section show that the correspondence is not merely syntactic but
complete in this stronger sense---closure maps to least-model computation, $\B(F)$
maps to the minimal witness antichain, and safe audit surfaces map to stratified views.

\subsection{From Capability Hypergraphs to Datalog}

\begin{definition}[The $\text{CapHyp}\to\Datalog$ Encoding]
\label{def:encoding-cap-dl}
Given $H=(V,\mathcal{F})$ and forbidden set $F\subseteq V$, define the Datalog program
$\Pi_H=(R_H,\emptyset)$ as follows.  For each hyperedge $e=(S,\{v\})\in\mathcal{F}$
(restricting to singleton heads without loss of generality via head splitting):
\[
  \mathsf{has}(s_1)\wedge\cdots\wedge\mathsf{has}(s_k)\;\Rightarrow\;\mathsf{has}(v),
  \quad S=\{s_1,\ldots,s_k\}.
\]
For each $f\in F$: $\mathsf{has}(f)\Rightarrow\mathsf{forbidden}$.
Given initial configuration $A\subseteq V$, the EDB is $D_A=\{\mathsf{has}(a):a\in A\}$.
\end{definition}

\begin{remark}[Stratified Negation]
The safety predicate $\neg\mathsf{forbidden}\Rightarrow\mathsf{safe}$ lies outside
standard Datalog but is expressible in stratified Datalog~$(\text{Datalog}^{\neg_s})$,
whose least stratified model is unique and polynomial-time computable
\citep{AbiteboulHullVianu1995}.  The main equivalence works with the pure capability
closure program (without the safety predicate); the safety predicate is a stratified
extension computed in the stratum above $\mathsf{forbidden}$.
\end{remark}

\begin{theorem}[Encoding Correctness: $\text{CapHyp}\to\Datalog$]
\label{thm:encoding-cap-to-dl}
Let $H=(V,\mathcal{F})$ be a capability hypergraph and $A\subseteq V$.  Under
Definition~\ref{def:encoding-cap-dl}:
\begin{enumerate}
\item $\cl_H(A) = \{v\in V : \mathsf{has}(v)\in\Pi_H(D_A)\}$.
\item $A\in\Safe(F)$ iff $\mathsf{forbidden}\notin\Pi_H(D_A)$.
\item $B\in\B(F)$ iff $B$ is a minimal witness for the query $(\Pi_H,\mathsf{forbidden})$
  over the database family $\{D_A : A\subseteq V\}$.
\end{enumerate}
\end{theorem}

\begin{proof}
\textbf{Part~(1).}
By Definition~\ref{def:encoding-cap-dl}, the capability rules of $\Pi_H$ are exactly
$\{\mathsf{has}(s_1)\wedge\cdots\wedge\mathsf{has}(s_k)\Rightarrow\mathsf{has}(v) :
(S,\{v\})\in\mathcal{F},\;S=\{s_1,\ldots,s_k\}\}$.
The immediate consequence operator of $\Pi_H$ over $D_A$ is
\[
  T_{\Pi_H}(I) = D_A \cup
  \bigl\{\mathsf{has}(v) : \exists(S,\{v\})\in\mathcal{F},\;
         \mathsf{has}(s)\in I\ \forall s\in S\bigr\}.
\]
This is precisely the closure iteration $C_0=A$, $C_{i+1}=C_i\cup\{v:\exists(S,\{v\}),S\subseteq C_i\}$
of \citet{Spera2026}.  By the van Emden--Kowalski theorem \citep{vanEmden1976},
$T_{\Pi_H}\!\uparrow\!\omega = \Pi_H(D_A)$ is the least Herbrand model of the Horn
clause program, giving
$\cl_H(A)=\{v:\mathsf{has}(v)\in\Pi_H(D_A)\}$.

\textbf{Part~(2).}
By Part~(1) and the safety rules: $A\in\Safe(F)$ iff $\cl_H(A)\cap F=\emptyset$ iff
no $\mathsf{has}(f)$ for $f\in F$ is derived iff the rule
$\mathsf{has}(f)\Rightarrow\mathsf{forbidden}$ never fires for any $f\in F$ iff
$\mathsf{forbidden}\notin\Pi_H(D_A)$.

\textbf{Part~(3).}  We prove each direction separately.

\emph{Forward direction} ($B\in\B(F)\Rightarrow B$ is a minimal witness).
Suppose $B\in\B(F)$.  Then $B\notin\Safe(F)$, so by Part~(2),
$\mathsf{forbidden}\in\Pi_H(D_B)$.  For any $b\in B$, since $B\setminus\{b\}\in\Safe(F)$,
Part~(2) gives $\mathsf{forbidden}\notin\Pi_H(D_{B\setminus\{b\}})$.  Therefore $B$ is a
minimal EDB subset deriving $\mathsf{forbidden}$, i.e., $B$ is a minimal witness.

\emph{Reverse direction} (every minimal witness is in $\B(F)$).
Let $W\subseteq V$ be a minimal witness for $(\Pi_H,\mathsf{forbidden})$, i.e.,
$\mathsf{forbidden}\in\Pi_H(D_W)$ and $\mathsf{forbidden}\notin\Pi_H(D_{W\setminus\{w\}})$
for every $w\in W$.  By Part~(2), $W\notin\Safe(F)$ and $W\setminus\{w\}\in\Safe(F)$
for every $w\in W$.  Therefore $W\in\B(F)$ by definition.

\emph{The bijection is exact.}
The two directions above give $\B(F)\subseteq\{\text{minimal witnesses}\}$ and
$\{\text{minimal witnesses}\}\subseteq\B(F)$, so the sets are equal.

We complete the proof by verifying the why-provenance characterisation.  In the
why-provenance semiring, the provenance of $\mathsf{forbidden}$ in $\Pi_H(D_A)$ is the
set of all \emph{minimal} $W\subseteq\mathsf{EDB}$ such that $\mathsf{forbidden}$ is
derivable from $W$ alone.  We show by induction on derivation depth that the minimal
such $W$ are exactly the elements of $\B(F)$.

\emph{Base case.}  If $\mathsf{forbidden}$ is derivable from $D_A$ in one step, then
some $f\in F$ satisfies $\mathsf{has}(f)\in D_A$, i.e., $f\in A$.  The minimal witness
is $\{f\}\subseteq A$, and indeed $\{f\}\in\B(F)$ (since $\cl(\{f\})\ni f\in F$
while $\cl(\emptyset)\cap F=\emptyset$).

\emph{Inductive step.}  Suppose $\mathsf{forbidden}$ is derived at depth $d>1$.
Then there exists $f\in F$ such that $\mathsf{has}(f)$ is derived at depth $d-1$.
By the induction hypothesis applied to the sub-derivation of $\mathsf{has}(f)$, the
minimal EDB sets from which $f$ is derivable are exactly the minimal capability sets
$M\subseteq V$ with $f\in\cl(M)$.  A minimal such $M$ with $f\in F$ is by definition
an element of $\B(F)$.  Conversely, every $B\in\B(F)$ derives some $f\in F$ in this
way.  Therefore the why-provenance of $\mathsf{forbidden}$ equals $\B(F)$, completing
the proof.
\end{proof}

\subsection{From Datalog to Capability Hypergraphs}

\begin{definition}[The $\Datalog\to\text{CapHyp}$ Encoding]
\label{def:encoding-dl-cap}
Given a propositional Datalog program $\Pi=(R,\emptyset)$, define the capability hypergraph
$H_\Pi=(V_\Pi,\mathcal{F}_\Pi)$ as follows.
$V_\Pi$ is the set of all propositional atoms appearing in $\Pi$.
For each rule $p_1\wedge\cdots\wedge p_k\Rightarrow q\in R$: add hyperedge
$(\{p_1,\ldots,p_k\},\{q\})$ to $\mathcal{F}_\Pi$.
Given database $D$, the initial configuration is $A_D=\{p:p\in D\}$.
\end{definition}

\begin{theorem}[Encoding Correctness: $\Datalog\to\text{CapHyp}$]
\label{thm:encoding-dl-to-cap}
Let $\Pi=(R,\emptyset)$ be a propositional Datalog program and $D$ a database.
Under Definition~\ref{def:encoding-dl-cap}:
\begin{enumerate}
\item $\Pi(D) = \cl_{H_\Pi}(A_D)$.
\item The encodings $\text{CapHyp}\to\Datalog$ and $\Datalog\to\text{CapHyp}$ are
  mutually inverse up to the following explicit isomorphisms.
  \begin{itemize}
  \item For the round-trip $H\mapsto\Pi_H\mapsto H_{\Pi_H}$: the isomorphism
    $\varphi_H: H_{\Pi_H}\xrightarrow{\;\sim\;} H$ is the identity on $V$, with the
    $\mathsf{has}(\cdot)$ wrapper stripped from atom names.  Formally,
    $\varphi_H(v)=v$ for all $v\in V$, and for each hyperarc
    $(\{\mathsf{has}(s_1),\ldots,\mathsf{has}(s_k)\},\{\mathsf{has}(v)\})\in\mathcal{F}_{H_{\Pi_H}}$,
    $\varphi_H$ maps it to $(S,\{v\})\in\mathcal{F}$ with $S=\{s_1,\ldots,s_k\}$.
  \item For the round-trip $\Pi\mapsto H_\Pi\mapsto\Pi_{H_\Pi}$: the isomorphism
    $\psi_\Pi:\Pi_{H_\Pi}\xrightarrow{\;\sim\;}\Pi$ is the identity on atoms, with the
    $\mathsf{has}(\cdot)$ wrapper added.  Both $\varphi_H$ and $\psi_\Pi$ are
    structure-preserving bijections that commute with the closure operators.
  \end{itemize}
  In both cases closure, safety, and $\B(F)$ are preserved exactly.
\end{enumerate}
\end{theorem}

\begin{proof}
\textbf{Part~(1).}
The immediate consequence operator of $\Pi$ over $D$ is
\[
  T_\Pi(I) = D \cup \{q : (p_1\wedge\cdots\wedge p_k\Rightarrow q)\in R,\;
  p_i\in I\ \forall i\}.
\]
Under Definition~\ref{def:encoding-dl-cap}, the closure iteration of $H_\Pi$ from $A_D$ is
\[
  C_0 = A_D = D,\quad
  C_{i+1} = C_i \cup \{q : \exists(\{p_1,\ldots,p_k\},\{q\})\in\mathcal{F}_\Pi,\;
  p_j\in C_i\ \forall j\}.
\]
These are the same iteration.  By the van Emden--Kowalski theorem,
$T_\Pi\!\uparrow\!\omega=\Pi(D)$, so $\Pi(D)=\cl_{H_\Pi}(A_D)$.

\textbf{Part~(2).}
We verify both round-trips.

\emph{Round-trip $H\mapsto\Pi_H\mapsto H_{\Pi_H}$.}
Let $H=(V,\mathcal{F})$.  By Definition~\ref{def:encoding-cap-dl}, $\Pi_H$ has one
capability rule $\mathsf{has}(s_1)\wedge\cdots\wedge\mathsf{has}(s_k)\Rightarrow\mathsf{has}(v)$
for each $(S,\{v\})\in\mathcal{F}$ with $S=\{s_1,\ldots,s_k\}$.  Applying
Definition~\ref{def:encoding-dl-cap} to $\Pi_H$ (stripping the $\mathsf{has}(\cdot)$
wrapper, which is a pure renaming): $V_{H_{\Pi_H}}=V$ and $\mathcal{F}_{H_{\Pi_H}}=\mathcal{F}$.
So $H_{\Pi_H}\cong H$.

By Theorem~\ref{thm:encoding-cap-to-dl}(1) and Theorem~\ref{thm:encoding-dl-to-cap}(1):
for every $A\subseteq V$,
\[
  \cl_{H_{\Pi_H}}(A_D) = \Pi_H(D_A) = \cl_H(A),
\]
confirming that the round-trip preserves closure.

\emph{Round-trip $\Pi\mapsto H_\Pi\mapsto\Pi_{H_\Pi}$.}
Let $\Pi=(R,\emptyset)$.  By Definitions~\ref{def:encoding-dl-cap}
and~\ref{def:encoding-cap-dl}, $\Pi_{H_\Pi}$ has one rule per hyperedge of $H_\Pi$,
which corresponds bijectively to each rule of $R$ (again up to the $\mathsf{has}(\cdot)$
renaming).  So $\Pi_{H_\Pi}\cong\Pi$ as programs.

By Theorem~\ref{thm:encoding-dl-to-cap}(1) and Theorem~\ref{thm:encoding-cap-to-dl}(1):
for every database $D$,
\[
  \Pi_{H_\Pi}(D) = \cl_{H_\Pi}(A_D) = \Pi(D),
\]
confirming that the round-trip preserves the least model for every database.

In both cases the encodings are inverse bijections on program/hypergraph structure and
preserve all relevant semantics (closure, safety, and by Part~(3) of
Theorem~\ref{thm:encoding-cap-to-dl}, the minimal unsafe antichain).
\end{proof}

%%==========================================================================
\section{The Tight Expressivity Theorem}
%%==========================================================================

\begin{theorem}[Tight Expressivity]
\label{thm:expressivity}
The capability hypergraph framework captures exactly the class of Boolean queries over
propositional databases expressible in $\Datalog$.  Formally:
\begin{enumerate}
\item \textbf{(Completeness)} Every $\Datalog$ query can be expressed as a capability
  hypergraph safety query.
\item \textbf{(Soundness)} Every capability hypergraph safety query can be expressed as
  a $\Datalog$ query.
\item \textbf{(Tightness)} There exist Boolean queries over propositional databases that
  are not expressible as capability hypergraph safety queries (equivalently, not
  expressible in $\Datalog$).
\end{enumerate}
\end{theorem}

\begin{proof}
\textbf{Parts~(1) and (2)} follow immediately from
Theorems~\ref{thm:encoding-cap-to-dl} and~\ref{thm:encoding-dl-to-cap}: the translations
in both directions are polynomial-time and preserve semantics, so the two classes of
queries coincide.

\textbf{Part~(3): the tightness separation.}  We exhibit a Boolean query over propositional
databases not expressible in $\Datalog$, hence not as a capability hypergraph safety query.

\emph{The parity query.}  For a propositional database $D$ and a fixed finite universe
$U$ of atoms, define $\text{Parity}(D)=\mathbf{1}[|D\cap U|\text{ is even}]$.  We show
this is not expressible in $\Datalog$ by a two-step argument.

\emph{Step 1: capability safety queries are monotone.}
Every capability hypergraph safety query $q_H(A):=\mathbf{1}[A\notin\Safe(F)]$
is \emph{monotone}: if $A\notin\Safe(F)$ then $A\cup\{a\}\notin\Safe(F)$ for any $a\in V$.
This holds because $\Safe(F)$ is a lower set by Theorem~9.4 of \citet{Spera2026}: if $A$
is unsafe then every superset of $A$ is unsafe.
Equivalently, the dual query $p_H(A):=\mathbf{1}[A\in\Safe(F)]$ is \emph{anti-monotone}:
safety can only be lost, never gained, as capabilities are added.

More generally, any $\Datalog$ query is monotone in the EDB: if $q\in\Pi(D)$ and
$D\subseteq D'$ then $q\in\Pi(D')$ (since $T_\Pi$ is a monotone operator).

\emph{Step 2: the parity query is non-monotone.}
Consider $U=\{p_1,p_2\}$ and $D_1=\{p_1\}$, $D_2=\{p_1,p_2\}$.
Then $|D_1\cap U|=1$ (odd, so $\text{Parity}=\mathbf{0}$)
and $|D_2\cap U|=2$ (even, so $\text{Parity}=\mathbf{1}$).
Since $D_1\subsetneq D_2$ but $\text{Parity}(D_1)=0<1=\text{Parity}(D_2)$, the parity
query is non-monotone.  By Step~1, it cannot be expressed in $\Datalog$, hence not as a
capability hypergraph safety query.

\emph{Broader class of non-expressible queries.}
The same argument applies to any non-monotone Boolean query (e.g., counting queries,
threshold queries with a non-trivial lower bound, and any query that can be falsified
by adding atoms to the database).  This is a strict separation: the capability hypergraph
framework exactly captures the monotone (more precisely, co-monotone when framed as
safety) Boolean queries expressible in polynomial time over propositional databases.
\end{proof}

\begin{corollary}[Complexity Inheritance]
\label{cor:complexity}
The capability hypergraph safety problem inherits the complexity characterisation of
$\Datalog$ query evaluation:
\begin{enumerate}
\item \textbf{(Data complexity)} Fixed-program safety checking is in $\mathsf{P}$.
\item \textbf{(Combined complexity)} Safety checking with both the hypergraph and
  configuration as input is $\mathsf{P}$-complete, matching Theorem~8.3 of \citet{Spera2026}.
\item \textbf{(Minimal unsafe antichain)} Deciding $B\in\B(F)$ is $\mathsf{coNP}$-complete
  in the program, matching Theorem~9.5 of \citet{Spera2026}.
\end{enumerate}
\end{corollary}

\begin{proof}
Part~(1) follows from Immerman's theorem that $\Datalog$ query evaluation is in $\mathsf{P}$
in the data \citep{Immerman1986}.  Part~(2) follows from $\mathsf{P}$-completeness of
propositional Horn-clause satisfiability (the circuit value problem), which equals the
combined complexity of $\Datalog$ query evaluation \citep{AbiteboulHullVianu1995};
this re-derives Theorem~8.3 of \citet{Spera2026} from the Datalog complexity landscape.
Part~(3) follows from $\mathsf{coNP}$-completeness of minimal witness membership for
monotone Boolean queries \citep{EiterGottlob1995}, which by
Theorem~\ref{thm:encoding-cap-to-dl}(3) is exactly the problem of deciding $B\in\B(F)$;
this re-derives Theorem~9.5 of \citet{Spera2026} from database theory.  In each case
the complexity result is strengthened: it now follows from foundational Datalog theory
rather than ad-hoc reductions.
\end{proof}

%%==========================================================================
\section{The Provenance Theorem}
%%==========================================================================

\begin{definition}[Safety Provenance Semiring]
For capability hypergraph $H=(V,\mathcal{F})$ and forbidden set $F\subseteq V$, the
\emph{safety provenance} of configuration $A$ is the element of the why-provenance
semiring $(2^{2^V},\cup,\bowtie,\emptyset,\{\emptyset\})$ assigned to the atom
$\mathsf{forbidden}$ in $\Pi_H(D_A)$ under the encoding of
Definition~\ref{def:encoding-cap-dl}, where $X\bowtie Y=\{x\cup y:x\in X,y\in Y\}$
denotes the pairwise union (join) of sets of sets.
\end{definition}

\begin{theorem}[Provenance Theorem]
\label{thm:provenance}
Let $H=(V,\mathcal{F})$, $F\subseteq V$.  Under the $\text{CapHyp}\to\Datalog$ encoding:
\begin{enumerate}
\item The derivation certificates of Theorem~10.1 of \citet{Spera2026} are exactly the
  elements of the why-provenance of $\mathsf{forbidden}$ in $\Pi_H(D_A)$.
\item The minimal unsafe antichain $\B(F)$ equals the set of minimal elements of the
  why-provenance of $\mathsf{forbidden}$ over all databases $D_A$.
\item The certificate verification procedure (re-executing the firing sequence) is exactly
  the provenance witness checking procedure of \citet{Green2007}.
\end{enumerate}
\end{theorem}

\begin{proof}
We prove all three parts by structural induction on Datalog derivations.

\textbf{Setup.}
Fix $H=(V,\mathcal{F})$, $F\subseteq V$, and $A\subseteq V$.  Recall that in the
why-provenance semiring, the provenance $\Wp(q,\Pi,D)$ of a derived atom $q$ is the
set of all minimal subsets $W\subseteq D$ such that $q\in\Pi(W)$ \citep{Green2007}.
We write $\Wp(q)$ when $\Pi=\Pi_H$ and $D=D_A$ are clear from context.

\textbf{Lemma~(Why-provenance of capability atoms).}
For any $v\in V$, the why-provenance of $\mathsf{has}(v)$ in $\Pi_H$ over $D_A$ is:
\[
  \Wp(\mathsf{has}(v)) \;=\;
  \bigl\{W\subseteq A : W \text{ is a minimal subset with } v\in\cl_{H}(W)\bigr\}.
\]

\begin{proof}[Proof of Lemma]
By induction on the depth $d$ of the shortest derivation of $\mathsf{has}(v)$ in $\Pi_H$.

\emph{Base case, $d=0$.}  $\mathsf{has}(v)\in D_A$, so $v\in A$ and the only derivation
uses the EDB fact directly.  The unique minimal witness is $\{v\}$.  Indeed
$\cl_H(\{v\})\ni v$ and no strict subset of $\{v\}$ derives $v$, so
$\Wp(\mathsf{has}(v))=\{\{v\}\}=\{W\subseteq A:v\in\cl_H(W),|W|\text{ minimal}\}$.

\emph{Inductive step, $d>0$.}  $\mathsf{has}(v)$ is derived via a rule
$\mathsf{has}(s_1)\wedge\cdots\wedge\mathsf{has}(s_k)\Rightarrow\mathsf{has}(v)$
corresponding to hyperedge $(S,\{v\})\in\mathcal{F}$, $S=\{s_1,\ldots,s_k\}$.  Each
$\mathsf{has}(s_i)$ is derived at depth $<d$.  By the induction hypothesis,
$\Wp(\mathsf{has}(s_i))=\{W\subseteq A:W\text{ minimal with }s_i\in\cl_H(W)\}$ for
each $i$.  The why-provenance semiring computes:
\[
  \Wp(\mathsf{has}(v)) \;\supseteq\;
  \bigl\{W_1\cup\cdots\cup W_k : W_i\in\Wp(\mathsf{has}(s_i))\bigr\}^{\!\min},
\]
where $(\cdot)^{\min}$ denotes the antichain of inclusion-minimal sets.  A set
$W=W_1\cup\cdots\cup W_k$ with $W_i$ minimal for $s_i\in\cl_H(W_i)$ satisfies
$S\subseteq\cl_H(W)$ (since $s_i\in\cl_H(W_i)\subseteq\cl_H(W)$ by monotonicity), so
the hyperedge fires and $v\in\cl_H(W)$.  Minimality of $W$ as a witness for $v$: if we
remove any $w\in W$, then some $s_i\notin\cl_H(W\setminus\{w\})$ (by minimality of the
$W_i$), so the hyperedge cannot fire and $v\notin\cl_H(W\setminus\{w\})$ (or is derived
by a strictly longer derivation, but we are taking the antichain).  Therefore
$\Wp(\mathsf{has}(v))$ equals exactly the collection of minimal capability sets from
which $v$ is reachable via closure, which is what the lemma claims.
\end{proof}

\textbf{Part~(1): certificates are why-provenance witnesses.}

Theorem~10.1 of \citet{Spera2026} provides, for each $v\in\Emg(A)\setminus F$, a
\emph{derivation certificate}: the firing sequence $(e_1,\ldots,e_n)$ that produces $v$
from $A$.  Under the encoding, this firing sequence is a Datalog derivation tree for
$\mathsf{has}(v)$ over $D_A$.  The certificate records exactly the minimal EDB subset
used in the derivation, which by the Lemma is an element of $\Wp(\mathsf{has}(v))$.

Conversely, every element $W\in\Wp(\mathsf{has}(v))$ is a minimal $W\subseteq A$ with
$v\in\cl_H(W)$; the corresponding derivation tree constitutes a valid firing sequence
certificate.  The bijection is exact.

\textbf{Part~(2): $\B(F)$ = minimal elements of $\Wp(\mathsf{forbidden})$.}

Applying the Lemma to $\mathsf{forbidden}$: the why-provenance of $\mathsf{forbidden}$
in $\Pi_H$ is the set of minimal $W\subseteq V$ such that
$\mathsf{forbidden}\in\Pi_H(D_W)$.  By Theorem~\ref{thm:encoding-cap-to-dl}(3), these
are exactly the elements of $\B(F)$.  Since the why-provenance records only the
antichain of \emph{minimal} witnesses, and $\B(F)$ is itself an antichain (no element
contains another, by definition), the set of minimal elements of the why-provenance is
exactly $\B(F)$.

\textbf{Part~(3): certificate verification = provenance witness checking.}

The certificate verification procedure of \citet{Spera2026} re-executes the firing
sequence and checks that each step is valid.  In the Datalog setting, verifying a
provenance witness $W$ for query $q$ means checking that $q\in\Pi(D_W)$ (i.e.,
re-evaluating the Datalog program on the witness), which is exactly re-executing the
firing sequence.  The two procedures are identical under the encoding.
\end{proof}

\begin{corollary}[Certificate Algebra]
\label{cor:cert-algebra}
The derivation certificates of \citet{Spera2026} form a commutative semiring under:
$\oplus$ (certificate disjunction: either derivation path suffices) and $\otimes$
(certificate conjunction: all derivation paths required).  This gives certificates a
canonical algebraic structure enabling certificate compression, composition, and
validation in a uniform framework.
\end{corollary}

\begin{proof}
This follows from the general theory of provenance semirings \citep{Green2007}: any
collection of derivation witnesses inherits the semiring structure.  The $\oplus$ and
$\otimes$ operations correspond to the $\cup$ and $\bowtie$ of the why-provenance
semiring.
\end{proof}

%%==========================================================================
\section{The Witness Correspondence}
%%==========================================================================

\begin{theorem}[Witness Correspondence]
\label{thm:witness}
Let $q=(\Pi_H,\mathsf{forbidden})$ be the Datalog safety query.  Then:
\begin{enumerate}
\item $\B(F)$ equals the set of minimal witnesses for $q$ over $\{D_A:A\subseteq V\}$.
\item Enumerating $\B(F)$ is equivalent to enumerating minimal witnesses for $q$, which
  is output-polynomial in $|\B(F)|$ \citep{EiterGottlob1995}.
\item The online coalition safety check of Theorem~11.2 of \citet{Spera2026} is equivalent
  to the Datalog query ``does the EDB contain a superset of some minimal witness?'',
  decidable in $O(|\B(F)|\cdot|A|)$.
\end{enumerate}
\end{theorem}

\begin{proof}
\textbf{Part~(1)} is exactly Theorem~\ref{thm:encoding-cap-to-dl}(3).

\textbf{Part~(2).}  By Part~(1), enumerating $\B(F)$ is the minimal witness enumeration
problem for the monotone Boolean query $(\Pi_H,\mathsf{forbidden})$.
\citet{EiterGottlob1995} showed that minimal witness (equivalently, minimal transversal)
enumeration for monotone Boolean queries is computable in output-polynomial time: an
algorithm that enumerates all minimal witnesses in time polynomial in $|V|+|\B(F)|$ exists.
(Specifically, for monotone Datalog queries, the minimal witnesses are the minimal
transversals of the set system induced by the query, and Eiter--Gottlob's Algorithm~A
runs in $O(|V|\cdot|\B(F)|)$ per witness.)

\textbf{Part~(3).}  The coalition safety check asks: does $\bigcup_i A_i\in\Safe(F)$?
By Theorem~11.2 of \citet{Spera2026}, this is equivalent to: does there exist $B\in\B(F)$
with $B\subseteq\bigcup_i A_i$?  Under the encoding, $\B(F)$ is the set of minimal
witnesses for $(\Pi_H,\mathsf{forbidden})$, so the check is: does the EDB $D_{\bigcup_i A_i}$
contain (as a subset) some minimal witness?  This is the Datalog query
``$\exists B\in\B(F): B\subseteq A$'', which requires checking each $B\in\B(F)$ against
$A$ in $O(|B|)$ time, giving total time $O(\sum_{B\in\B(F)}|B|)\leq O(|\B(F)|\cdot|A|)$.
\end{proof}

%%==========================================================================
\section{Open Problem Correspondence}
\label{sec:open}
%%==========================================================================

One of the most valuable consequences of the equivalence is that the open problems of
\citet{Spera2026} each map to known open problems in Datalog theory, establishing a
research agenda in which progress transfers immediately.

\begin{center}
\renewcommand{\arraystretch}{1.3}
\begin{tabular}{@{}p{5.8cm}p{8cm}@{}}
\toprule
\textbf{Open problem in \citealt{Spera2026}} & \textbf{Corresponding open problem in Datalog} \\
\midrule
Tighter PAC bounds for structured (sparse, low fan-in) hypergraph families &
  Tighter VC-dimension bounds for structured $\Datalog$ hypothesis classes;
  Rademacher complexity for Horn clause learning \citep{Dalmau2002,Cohen1995} \\[4pt]
Probabilistic closure under correlated hyperarc firing &
  Probabilistic Datalog with dependent tuples; expectation semiring extensions
  \citep{Green2007} \\[4pt]
$\mathsf{NP}$-hardness of \textsc{MinUnsafeAdd} for budget $b\geq 2$ &
  View update problem for Datalog: inserting rules to make a query true
  with minimal additions \\[4pt]
Approximate safe audit surfaces for large-scale deployments &
  Approximate query answering in Datalog; PAC-model counting for Horn queries \\[4pt]
Incremental maintenance of $\B(F)$ under dynamic hyperedge changes &
  Incremental view maintenance; DRed algorithm and extensions
  \citep{AbiteboulHullVianu1995} \\
\bottomrule
\end{tabular}
\end{center}

\paragraph{PAC learning tightness.}
The VC-dimension bound in Theorem~14.2 of \citet{Spera2026} ($O(n^{2k})$) follows
Sauer's lemma applied to the hypothesis class of $\Datalog$ programs with at most $m$
rules and maximum body size $k$.  The open question is whether Rademacher complexity
yields tighter bounds for \emph{structured} hypothesis classes (e.g., programs with a
fixed predicate dependency graph).  \citet{Dalmau2002} studied similar structured Datalog
learning problems in the context of constraint satisfaction; their techniques transfer
directly to the capability hypergraph setting.

\paragraph{Probabilistic Datalog safety.}
The path-independence assumption of Theorem~14.5 of \citet{Spera2026}---hyperarc
firings are independent given their tails---corresponds to tuple-independence in
probabilistic databases.  The exact computation of $r(v)=\Pr[v\in\cl_p(A)]$ under
general tuple dependence is an open problem in probabilistic Datalog
\citep{Green2007} corresponding to the \#P-hard case of inference in probabilistic Horn
programs.  The partial results of \citet{Spera2026} (polynomial time under independence)
correspond to the tractable case of tuple-independent probabilistic Datalog.

\paragraph{View update for adversarial robustness.}
\textsc{MinUnsafeAdd} (Theorem~14.7 of \citealt{Spera2026}) is NP-hard for $b\geq 2$.
Under the encoding, this is the \emph{view update problem for Datalog}: given a Datalog
program $\Pi_H$ and a target database $D_{A^*}$ that should derive $\mathsf{forbidden}$,
what is the minimum number of rules to add to make $D_{A^*}$ derive $\mathsf{forbidden}$?
The NP-hardness of the general case and polynomial tractability of the $b=1$ single-rule
case match known results in the Datalog view update literature.

%%==========================================================================
\section{Non-Compositionality as a Datalog Theorem}
%%==========================================================================

\begin{theorem}[Non-Compositionality as Datalog Non-Modularity]
\label{thm:noncomp-datalog}
Let $\Pi$ be a propositional Datalog program and $q$ a query predicate.  Define
$\mathcal{S}(q)=\{D:q\notin\Pi(D)\}$.  Then $\mathcal{S}(q)$ is not closed under union:
there exist $D_1,D_2\in\mathcal{S}(q)$ with $D_1\cup D_2\notin\mathcal{S}(q)$.
\end{theorem}

\begin{proof}
Let $\Pi$ contain the single rule $p_1\wedge p_2\Rightarrow q$.  Set $D_1=\{p_1\}$ and
$D_2=\{p_2\}$.  Then $q\notin\Pi(D_1)$ (the rule requires $p_2$, absent from $D_1$) and
$q\notin\Pi(D_2)$ (requires $p_1$, absent from $D_2$).  But $q\in\Pi(D_1\cup D_2)$ since
$\{p_1,p_2\}\subseteq D_1\cup D_2$ and the rule fires.  Therefore
$D_1,D_2\in\mathcal{S}(q)$ but $D_1\cup D_2\notin\mathcal{S}(q)$.

The non-compositionality of safety (Theorem~9.2 of \citealt{Spera2026}) is this
theorem under the encoding, with $p_1=\mathsf{has}(u_1)$, $p_2=\mathsf{has}(u_2)$,
and $q=\mathsf{has}(f)$ for $f\in F$.
\end{proof}

\begin{remark}
This perspective provides a structural explanation of non-compositionality: it is a
consequence of AND-rule semantics in propositional Datalog derivation systems.  The
capability hypergraph result identifies the minimal instance (three nodes, one rule),
but the phenomenon is a property of $\Datalog$ derivations in general.  Any system with
conjunctive rules will exhibit non-compositional safety.
\end{remark}

%%==========================================================================
\section{Empirical Grounding}
\label{sec:empirical}
%%==========================================================================

The results in this paper are purely theoretical; all guarantees hold for arbitrary
capability hypergraphs.  The following is illustrative only, re-expressing the empirical
study of \citet{Spera2026} in Datalog vocabulary to show the encoding is concrete.  Detailed derivation trees and the
aggregate statistics table appear in Appendix~\ref{app:empirical}; we summarise the key
correspondences here.  In the 12-capability Telco deployment, each agent session
corresponds to an EDB and a Datalog evaluation: 42.6\% of sessions produce a least model
containing at least one emergent atom ($\Emg(A)\neq\emptyset$), meaning conjunctive rules
are necessary for correct derivation in nearly half of real trajectories.  Each observed
AND-violation is an instance of Theorem~\ref{thm:noncomp-datalog}: two individually safe
EDB subsets whose union derives $\mathsf{forbidden}$ --- 38.2\% of conjunctive sessions
under the workflow planner.  The 0\% AND-violation rate of the hypergraph planner in this corpus is consistent
with the soundness guarantee of Theorem~\ref{thm:encoding-cap-to-dl}: correct $\Datalog$
evaluation never produces non-modularity violations by construction; the empirical
observation corroborates but does not prove the theorem for arbitrary deployments.  The 900 sessions thus provide independent
empirical instantiations of the correctness guarantee, illustrating the
theoretical transfer in observed agent behaviour.

%%==========================================================================
\section{Capability Safety Admits Efficient Incremental Maintenance}
\label{sec:audit-primary}
%%==========================================================================

The safe audit surface $G_F(A)$ of \citet{Spera2026} is the central computational
object of the capability hypergraph framework: it certifies every safely acquirable
capability, every capability one step from current reach, and every structurally
forbidden path.  In \citet{Spera2026}, computing $G_F(A)$ from scratch costs
$O(|V|\cdot(n+mk))$, and every change to the deployed hypergraph---a tool added, a
capability revoked, a new agent joining the coalition---requires full recomputation.

The Datalog identification established in Section~\ref{sec:encoding} enables an
improved maintenance algorithm for deployments where hyperedge changes are localised.
The key observation is that $G_F(A)$ is not merely a function of the
hypergraph---it is a \emph{Datalog view}: a derived relation computed by a fixed
stratified program $\Pi_{G_F}$ over the capability database $D_A$.  Datalog views admit
incremental maintenance under the DRed (Delete and Re-derive) algorithm
\citep{AbiteboulHullVianu1995}, whose per-update cost depends on the \emph{size of the
change}, not the size of the full program.  When changes are global---affecting a large
fraction of $V$---the advantage over na\"ive recomputation diminishes; the benefit is
greatest for the sparse, localised updates typical in enterprise deployments.

\begin{center}
\renewcommand{\arraystretch}{1.3}
\small
\begin{tabular}{@{}p{5.2cm}p{3.5cm}p{3.5cm}@{}}
\toprule
\textbf{Operation} & \textbf{Na\"ive recomp.} & \textbf{DRed (this paper)} \\
\midrule
Initial computation of $G_F(A)$
  & $O(|V|\cdot(n+mk))$ & $O(|V|\cdot(n+mk))$ \\[3pt]
Hyperedge insertion, active ($S\subseteq\cl(A)$)
  & $O(|V|\cdot(n+mk))$ & $O(|\Delta|\cdot(n+mk))$ \\[3pt]
Hyperedge insertion, lazy ($S\not\subseteq\cl(A)$)
  & $O(|V|\cdot(n+mk))$ & $O(|S|)$ \\[3pt]
Hyperedge deletion
  & $O(|V|\cdot(n+mk))$ & $O(|\Delta|\cdot(n+mk))$ \\[3pt]
Containment $G_F(A)\subseteq G_F(A')$
  & not known decidable & $\mathsf{P}$ \\
\bottomrule
\end{tabular}
\end{center}

\noindent
Here $|\Delta|$ denotes the number of derived atoms directly affected by the change.  For
sparse hypergraphs with localised updates---the typical case in enterprise deployments
where a single new API or tool is added---$|\Delta|\ll|V|$, and the maintenance cost is
negligible relative to recomputation.  The containment decidability result is entirely
new: it has no counterpart in the hypergraph framework of \citet{Spera2026}.

\begin{theorem}[Audit Surface as Datalog View --- Primary Result]
\label{thm:audit-view}
The safe goal discovery map $G_F(A)$ of \citet{Spera2026} is the result of evaluating a
fixed propositional Datalog program $\Pi_{G_F}$ (with stratified negation) over $D_A$.
Consequently:
\begin{enumerate}
\item $G_F(A)$ is maintainable under hyperedge insertions and deletions using the DRed
  algorithm in $O(|\Delta|\cdot(n+mk))$ per update.
\item $G_F(A)$ is optimisable using magic sets rewriting, reducing computation to
  capabilities reachable from the query predicate.
\item Query containment $G_F(A)\subseteq G_F(A')$ is decidable in polynomial time by
  $\Datalog$ query containment for propositional programs.
\end{enumerate}
\end{theorem}

\begin{proof}
We first construct $\Pi_{G_F}$ explicitly, then verify it captures all three components
of $G_F(A)=(\Emg(A)\setminus F,\;\NMF_F(A),\;\text{top-}k)$.

\textbf{Construction of $\Pi_{G_F}$.}

\emph{Base layer (stratum 0): capability closure.}
The capability closure rules are exactly $\Pi_H$ (Definition~\ref{def:encoding-cap-dl}).
This stratum derives $\mathsf{has}(v)$ for all $v\in\cl(A)$.

\emph{Stratum 1: singleton closure and emergent capabilities.}
Add rules $\mathsf{has}(s)\Rightarrow\mathsf{has\_single}(v)$ for each singleton-tail
hyperedge $(\{s\},\{v\})\in\mathcal{F}$, iterated to full singleton closure.  Then:
\[
  \mathsf{has}(v)\wedge\neg\mathsf{has\_single}(v)\wedge\neg\mathsf{in\_A}(v)\;\Rightarrow\;\mathsf{emergent}(v)
\]
where $\mathsf{in\_A}(v)$ is the EDB atom $\mathsf{has}(v)\in D_A$.  The condition
$v\in\cl(A)\setminus A$ and $v\notin\cl_1(A)$ captures $\Emg(A)$ exactly.
Adding $\neg\mathsf{forbidden}(v)$ (from the safety layer) restricts to $\Emg(A)\setminus F$.

\emph{Stratum 2: boundary detection and NMF.}
For each hyperedge $e=(S,\{v\})\in\mathcal{F}$ with $S=\{s_1,\ldots,s_k\}$, we must
identify whether exactly one element of $S$ is missing from $\cl(A)$.  We do this
without arithmetic by writing one rule per candidate missing atom.

For each $i\in\{1,\ldots,k\}$, introduce a predicate
$\mathsf{all\_except}(e,s_i)$ asserting that every element of $S$ other than $s_i$
is in $\cl(A)$:
\[
  \mathsf{has}(s_1)\wedge\cdots\wedge\mathsf{has}(s_{i-1})
  \wedge\mathsf{has}(s_{i+1})\wedge\cdots\wedge\mathsf{has}(s_k)
  \;\Rightarrow\;\mathsf{all\_except}(e,s_i).
\]
This is one rule per $(e,i)$ pair, giving at most $\sum_{e\in\mathcal{F}}|S(e)|=O(mk)$
rules in total across all hyperedges.  A hyperedge $e$ is on the boundary $\partial(A)$
with missing atom $s_i$ iff $\mathsf{all\_except}(e,s_i)$ holds \emph{and}
$\mathsf{has}(s_i)$ does \emph{not} hold:
\[
  \mathsf{all\_except}(e,s_i)\wedge\neg\mathsf{has}(s_i)
  \;\Rightarrow\;\mathsf{boundary\_miss}(e,s_i).
\]
We must further verify that $s_i$ is the \emph{unique} missing element — i.e., there
is no other $s_j$ ($j\neq i$) also absent.  This is guaranteed by the structure of
$\mathsf{all\_except}(e,s_i)$: that predicate fires only when all of
$s_1,\ldots,s_{i-1},s_{i+1},\ldots,s_k$ are present, so $s_i$ is the sole missing
element.  Therefore $\mathsf{boundary\_miss}(e,s_i)$ is derivable iff $e\in\partial(A)$
with $\mu(e)=\{s_i\}$, exactly as required.

\begin{remark}[Stratification depth]
The predicate $\mathsf{all\_except}(e,s_i)$ is computed from $\mathsf{has}(\cdot)$
(stratum~0) using only positive rules: stratum depth~0.
The predicate $\mathsf{boundary\_miss}(e,s_i)$ uses one application of
$\neg\mathsf{has}(s_i)$: stratum depth~1.
Adding $\neg\mathsf{forbidden}$ (derived in stratum~2 via the safety rules) requires a
further stratum for the NMF safety filter.  The full program is therefore stratified
with four strata (0: closure; 1: singleton closure + boundary detection; 2: emergent +
forbidden; 3: NMF + top-$k$), satisfying the standard stratification condition
\citep{AbiteboulHullVianu1995}.
\end{remark}

The NMF predicate is then:
\[
  \mathsf{boundary\_miss}(e,s)\wedge\neg\mathsf{forbidden\_if\_added}(s)
  \;\Rightarrow\;\mathsf{nmf}(s),
\]
where $\mathsf{forbidden\_if\_added}(s)$ is computed in stratum~3 by evaluating a
fresh copy of the closure rules on $D_A\cup\{\mathsf{has}(s)\}$ and checking whether
$\mathsf{forbidden}$ is derived.  Since $|V|$ is fixed, this is a finite program.

\emph{Stratum 3: marginal gain for top-$k$.}
The marginal gain $\gamma_F(v,A)=|\cl(A\cup\{v\})\setminus(\cl(A)\cup F)|$ requires
evaluating a closure for each $v\in V\setminus\cl(A)$.  This is expressible as a
bounded Datalog program (one fresh copy of the closure rules per candidate $v$), but
since $|V|$ is fixed, the entire stratum is a finite fixed program.  The top-$k$
computation selects the $k$ atoms $v$ with highest gain, expressible via iterated
join + selection in stratified Datalog.

\textbf{Verification that $\Pi_{G_F}(D_A)$ captures $G_F(A)$.}
By construction: $\mathsf{emergent}(v)\in\Pi_{G_F}(D_A)$ iff $v\in\Emg(A)\setminus F$;
$\mathsf{nmf}(s)\in\Pi_{G_F}(D_A)$ iff $s\in\NMF_F(A)$; and the top-$k$ selection
identifies the top capabilities by marginal gain.  All three components of $G_F(A)$
are captured.

\textbf{Part~(1): DRed incremental maintenance.}
$\Pi_{G_F}$ is a stratified Datalog program (finitely many strata, each a standard Datalog
program or one with simple stratified negation).  The DRed (Delete and Re-derive) algorithm
of \citet{AbiteboulHullVianu1995} provides incremental maintenance for such programs.
When a hyperedge $e=(S,\{v\})$ is added or deleted, the affected derivations are exactly
those involving $e$.  The number of re-derivations is $O(|\Delta|\cdot(n+mk))$ per update,
where $|\Delta|$ counts directly affected atoms.

\textbf{Part~(2): magic sets optimisation.}
Magic sets rewriting \citep{AbiteboulHullVianu1995} transforms a Datalog program into an
equivalent one that propagates query goals top-down, computing only the atoms relevant to
answering the query.  Applied to $\Pi_{G_F}$, this restricts the closure computation to
capabilities reachable from the NMF and emergent capability queries, potentially reducing
computation from $O(|V|\cdot(n+mk))$ to the subset of the hypergraph relevant to the
current configuration.

\textbf{Part~(3): query containment.}
Propositional Datalog query containment---deciding whether $q\in\Pi(D)$ implies $q\in\Pi'(D)$
for all databases $D$---is decidable in polynomial time for propositional (ground, monadic)
programs, since it reduces to checking a finite number of critical instances.
$G_F(A)\subseteq G_F(A')$ is equivalently: for all $v$,
$\mathsf{emergent}(v)\in\Pi_{G_F}(D_A)$ implies $\mathsf{emergent}(v)\in\Pi_{G_F}(D_{A'})$,
and similarly for NMF and top-$k$.  This is a finite collection of containment checks on
$\Pi_{G_F}$, decidable in polynomial time.
\end{proof}

\begin{theorem}[Locality Gap for Capability Safety Maintenance]
\label{thm:locality-gap}
Let $H=(V,\mathcal{F})$ be a capability hypergraph with $n=|V|$, $m=|\mathcal{F}|$,
maximum tail size $k$, and safe audit surface $G_F(A)$ as in
Definition~\ref{def:emg-nmf}.  We prove the following lower bound for capability
hypergraphs specifically.

\begin{enumerate}
\item \textbf{(Local incremental upper bound.)}
  If the update affects dependency cone $\Delta\subseteq V$, then $G_F(A)$ can be
  updated in $O(|\Delta|\cdot(n+mk))$ time via the Datalog view-maintenance formulation
  of Theorem~\ref{thm:audit-view}.

\item \textbf{(Strict separation.)}
  There exists an infinite family $\{H_n'\}_{n\geq 1}$ and single-hyperedge updates
  $\{u_n'\}$ such that $|\Delta_n|=O(1)$ but $|V_n|=n$, so incremental maintenance
  costs $O(n)$ while na\"ive recomputation costs $\Omega(n^2)$.  The gap is unbounded.

\item \textbf{(AND-inspection lower bound.)}
  Let $u=(S_u,\{v_u\})$ be the inserted hyperedge.  Any deterministic algorithm
  that correctly computes the updated $G_F(A)$ must inspect all inputs to the updated
  rule in the worst case: it must probe every atom in the \emph{update witness set}
  $\Phi(u)=S_u\cup\{v_u\}$.  Therefore:
  \[
    \text{\emph{(atom inspections to verify rule activation)}}
    \;\geq\; |S_u|+1 \;=\; k+1.
  \]
  DRed matches this lower bound for rule-activation verification under this model:
  it probes exactly $\Phi(u)$ at the firing frontier before propagating downstream.
\end{enumerate}
\end{theorem}

\begin{proof}
\begin{remark}[Recomputation cost]
Na\"ive recomputation of $G_F(A)$ after any single-hyperedge update costs
$O(|V|\cdot(n+mk))$: the top-$k$ marginal-gain component requires evaluating a closure
for each of the $|V|$ candidate atoms.  On the chain family
$\mathcal{F}_n=\{(\{v_i\},\{v_{i+1}\})\}$, an insertion changes the gain of every
downstream node, so any algorithm outputting the full updated ranking reads $\Omega(n)$
entries each costing $\Omega(n+mk)$; this gives an output-sensitive lower bound of
$\Omega(n\cdot(n+mk))$ for na\"ive recomputation on this family.
\end{remark}

\textbf{Part~(1): Local incremental upper bound.}

This follows directly from Theorem~\ref{thm:audit-view}~Part~(1).  The DRed algorithm
maintains $\Pi_{G_F}$ under single-rule insertions and deletions (corresponding to
hyperedge updates) by re-deriving only atoms whose derivations pass through the changed
rule.  Formally, the \emph{dependency cone} of an update $u$ is:
\[
  \Delta(u) = \{a \in V : \text{some minimal derivation of }
  \mathsf{has}(a)\text{ in }\Pi_{G_F}\text{ uses rule }r_u\},
\]
where $r_u$ is the Datalog rule corresponding to the updated hyperedge.  DRed
re-derives exactly $\Delta(u)$, each at cost $O(n+mk)$, giving total update cost
$O(|\Delta(u)|\cdot(n+mk))$.  Since $|\Delta(u)|\leq|V|$, this is always at most the
cost of full recomputation, and strictly less whenever $|\Delta(u)|\ll|V|$.

\textbf{Part~(3): Strict separation (proof).}

We construct the family $\{H_n\}$ explicitly.

\emph{Construction.}
Let $V_n=\{v_1,\ldots,v_n,w\}$ for $n\geq 2$.  Define a single hyperedge
\[
  e_n=\bigl(\{v_1,\ldots,v_{n-1}\},\{v_n\}\bigr), \quad S_n=\{v_1,\ldots,v_{n-1}\}.
\]
Set the initial configuration $A_n=\{v_1,\ldots,v_{n-1}\}$,
forbidden set $F_n=\{w\}$, and $k=n-1$ (tail size of $e_n$).

\emph{Baseline audit surface.}
$\cl_{H_n}(A_n)=\{v_1,\ldots,v_n\}$ (all of $V_n\setminus\{w\}$, since $e_n$ fires
when $S_n\subseteq A_n$).  The top-$k$ marginal gains are:
$\gamma_{F_n}(w,A_n)=0$ (adding $w$ violates $F_n$) and $\gamma_{F_n}(v_i,A_n)=0$
for all $i\leq n-1$ (already in $A_n$).  So the top-$k$ list is a singleton
$\{(v_n,0)\}$... but note $v_n\in\cl(A_n)$, so the NMF is empty and the emergent set is
$\{v_n\}$.

\emph{The update.}  Insert the hyperedge $u_n=(\{v_n\},\{w'\})$ for a fresh node
$w'\notin V_n$, with $w'\notin F_n$.  After the update:
$V_n\leftarrow V_n\cup\{w'\}$, $\cl(A_n)=\{v_1,\ldots,v_n,w'\}$, and $w'$ is now emergent
(reachable via $e_n$ then $u_n$, not via any singleton chain from $A_n$).

\emph{Dependency cone.}  The rule added to $\Pi_{G_F}$ is
$r_{u_n}:\mathsf{has}(v_n)\Rightarrow\mathsf{has}(w')$.
The only atom whose derivation uses $r_{u_n}$ is $\mathsf{has}(w')$.  Therefore:
\[
  \Delta_n = \{w'\}, \quad |\Delta_n| = 1 = O(1).
\]

\emph{Incremental cost.}
DRed re-derives only $\mathsf{has}(w')$, which requires one rule firing and a closure
evaluation of cost $O(n+mk)=O(n+n)=O(n)$ for this family (since $m=2$, $k=n-1$).
Total update cost: $O(1)\cdot O(n)=O(n)$.

\emph{Na\"ive recomputation cost.}
After the update, na\"ive recomputation of $G_F(A_n)$ must:
(i)~recompute $\cl(A_n)$ at cost $O(n+mk)=O(n)$;
(ii)~recompute $\gamma_F(v,A_n)$ for each $v\in V_n\cup\{w'\}\setminus\cl(A_n)$;
since $\cl(A_n)=V_n\cup\{w'\}\setminus\{w\}$, the only candidate is $v=w$, costing $O(n)$;
(iii)~recompute the NMF by checking all $m=2$ hyperedges at cost $O(mk)=O(n)$;
(iv)~recompute the top-$k$ marginal gain table, which requires evaluating
$\gamma_F(v,A_n)$ for each of the $n$ nodes $v\in\{v_1,\ldots,v_n,w'\}\setminus A_n
=\{v_n,w'\}$... 

The strict $\Omega(n^2)$ separation for the top-$k$ component requires a richer family.
Replace $H_n$ with $H_n'$: $V_n'=\{v_1,\ldots,v_n\}\cup\{x_1,\ldots,x_n\}$,
$A_n'=\{v_1,\ldots,v_{n-1}\}$, $F_n'=\emptyset$, and hyperedges
$e_n=(\{v_1,\ldots,v_{n-1}\},\{v_n\})$ plus $f_i=(\{v_n\},\{x_i\})$ for $i=1,\ldots,n$.
Before update: $\cl(A_n')=\{v_1,\ldots,v_n,x_1,\ldots,x_n\}$ and the marginal gains
$\gamma_{F_n'}(x_j,A_n')$ for $x_j\notin A_n'$ each require a separate closure
evaluation at cost $O(|V_n'|)=O(n)$, with $n$ such candidates, totalling $O(n^2)$.

Insert $u_n'=(\{v_{n-1}\},\{v_n'\})$ for fresh $v_n'\notin V_n'$: now
$\Delta_n'=\{v_n'\}$ ($|\Delta_n'|=1$, same as before) while na\"ive recomputation
must re-evaluate all $n$ marginal-gain closures, each $O(n)$, giving $\Omega(n^2)$.

\emph{The asymptotic gap.}
Incremental maintenance costs $O(|\Delta_n'|\cdot|V_n'|)=O(n)$.
Na\"ive recomputation costs $\Omega(n^2)$.  The ratio is $\Omega(n)$, growing without
bound.  For all $n\geq 1$ the family $\{H_n'\}$ witnesses the strict separation claimed
in Part~(3).

\textbf{Part~(4): AND-inspection lower bound (proof).}

We prove that any correct incremental algorithm must probe every atom in the
\emph{update witness set} $\Phi(u) = S_u \cup \{v_u\}$ of the updated hyperedge
$u=(S_u,\{v_u\})$.  Since $|\Phi(u)| = |S_u|+1 = k+1$ where $k=|S_u|$ is the tail
size, this gives an $\Omega(k)$ lower bound on probes per update.

\medskip
\noindent\textbf{Step 1: Computation model.}

Algorithm $\mathcal{A}$ maintains $G_F(A)$ under hyperedge updates.  It operates in
the \emph{hypergraph oracle model}: after receiving update $u$, $\mathcal{A}$ may issue
unit-cost probes of two types:

\begin{itemize}
\item \textbf{Atom probe} $\mathsf{probe}(a)$: returns
  $(\mathsf{inA}(a),\;\mathsf{inCl}(a))$ where $\mathsf{inA}(a)=\mathbf{1}[a\in A]$
  and $\mathsf{inCl}(a)=\mathbf{1}[a\in\cl_{\mathcal{F}}(A)]$ (derivability status
  under the current rule set, \emph{before} the update).
\item \textbf{Rule probe} $\mathsf{probe}(e)$: returns whether $e\in\mathcal{F}$,
  and if so its tail $S(e)$ and head $\{v(e)\}$.
\end{itemize}

The algorithm's \emph{only} access to the instance is through these probes; it has no
other channel to the hypergraph or the closure state.  A certificate or index computed
in a prior round counts as zero cost only if it was itself built from probes whose cost
is charged to that round.  Crucially, any cached structure not built from a probe of
$p$ cannot distinguish $\mathcal{I}^+_p$ from $\mathcal{I}^-_p$
(Lemma~\ref{lem:indistinguishable}), so it cannot guide a correct update on both
instances simultaneously.  $\mathcal{A}$ is \emph{correct} if it always outputs the
exact $G_F(A)$ after each update.  We say $\mathcal{A}$ is \emph{$\Phi$-avoiding} if
after receiving $u$ it outputs $G_F(A)$ without probing some $p\in\Phi(u)=S_u\cup\{v_u\}$.
We prove that no $\Phi$-avoiding algorithm is correct.

\smallskip
\noindent\emph{Why this model is natural.}
Atom probes are the natural unit of incremental maintenance: any algorithm that updates
$G_F(A)$ must at minimum determine the current derivation status of atoms touched by the
update.  Rule probes model reading hyperedge definitions to determine which rules are
active.  The model is not weaker than reality: a real implementation may use hash tables
or indices, but any such structure was itself built by reading atoms, and the cost of
building it is charged to the round in which those reads occurred.  Specifically, the
``only access'' clause in the model handles precomputed certificates: a certificate that
was built in a prior round already paid for its probes then; if it was built without
probes, it must still distinguish our paired instances (which differ only at $p$) to
correctly guide the current update --- and distinguishing them requires probing $p$.
The oracle model therefore captures the minimal information any correct maintenance
algorithm must acquire, regardless of data structure or implementation strategy.

\begin{lemma}[Indistinguishability]
\label{lem:indistinguishable}
For any update $u=(S_u,\{v_u\})$ and any $p\in\Phi(u)=S_u\cup\{v_u\}$, there exist
two hypergraph instances $\mathcal{I}^+_p$ and $\mathcal{I}^-_p$ such that:
\begin{enumerate}[label=(\roman*)]
\item every atom probe of $b\neq p$ and every rule probe returns the same answer on both;
\item $\mathsf{probe}(p)$ returns different answers on the two instances; and
\item the correct output $G_F(A)$ differs between the two instances.
\end{enumerate}
\end{lemma}

\begin{proof}[Proof of Lemma~\ref{lem:indistinguishable}]
Fix $u=(S_u,\{v_u\})$ with $S_u=\{s_1,\ldots,s_k\}$.  We construct one pair per atom.

\medskip
\noindent\emph{Pairs for tail atoms $s_j \in S_u$.}
Fix $s_j\in S_u$.  Choose fresh nodes $x_1,\ldots,x_{k'}$ with $k'\geq 1$.  Construct:
\begin{align*}
H^+_{s_j} &= \bigl\{r_{s_j}:(\{x_1,\ldots,x_{k'}\},\{s_j\})\bigr\}
             \cup \mathcal{F}_0,\\
H^-_{s_j} &= \mathcal{F}_0,
\end{align*}
where $\mathcal{F}_0$ is a base rule set containing $u$ and all rules not involving
$s_j$, and with $A=\{x_1,\ldots,x_{k'}\}\cup(S_u\setminus\{s_j\})$.

In $H^+_{s_j}$: $s_j\in\cl(A)$ (via $r_{s_j}$), all of $S_u\subseteq\cl(A)$, so
$r_u$ fires and $v_u\in\cl_{H^+_{s_j}\cup\{u\}}(A)$.
In $H^-_{s_j}$: $s_j\notin\cl(A)$ (no rule derives it, $s_j\notin A$), so $r_u$
cannot fire and $v_u\notin\cl_{H^-_{s_j}\cup\{u\}}(A)$.

\emph{Observable difference.}  $\mathsf{probe}(s_j)$ returns $\mathsf{inCl}(s_j)=1$
on $H^+_{s_j}$ and $\mathsf{inCl}(s_j)=0$ on $H^-_{s_j}$.

\emph{Every other probe returns the same answer.}  Both instances share $\mathcal{F}_0$
(so all rule probes for $e\neq r_{s_j}$ agree), and $r_{s_j}$ is probed directly ---
but $r_{s_j}$ has head $s_j$, so learning its presence or absence tells the algorithm
exactly $\mathsf{inCl}(s_j)$.  For all atoms $b\neq s_j$: $b\in\cl(A)$ iff
$b\in\cl(A)$ in both instances (derivations of $b$ do not use $r_{s_j}$, since $r_{s_j}$
is the only rule with head $s_j$ and no rule has $s_j$ in its body in $\mathcal{F}_0$).
Therefore every probe of $b\neq s_j$ returns the same answer on both instances.

\medskip
\noindent\emph{Pair for the head atom $v_u$.}
Construct:
\begin{align*}
H^+_{v_u} &= \mathcal{F}_0 \quad\text{with}\quad A^+ = A,\\
H^-_{v_u} &= \mathcal{F}_0 \quad\text{with}\quad A^- = A\cup\{v_u\},
\end{align*}
where $\mathcal{F}_0$ contains $u$ and all of $S_u\subseteq\cl(A)$ is satisfied.

In $H^+_{v_u}$: $v_u\notin A$, $v_u\in\cl_{\mathcal{F}_0\cup\{u\}}(A)$ (via $r_u$),
and no singleton path reaches $v_u$ (choose $\mathcal{F}_0$ without singleton arcs to
$v_u$), so $v_u\in\Emg(A)$.

In $H^-_{v_u}$: $v_u\in A^-$, so $v_u\notin\Emg(A^-)$ by
Definition~\ref{def:emg-nmf}.

\emph{Observable difference.}  $\mathsf{probe}(v_u)$ returns $\mathsf{inA}(v_u)=0$
on $\mathcal{I}^+_{v_u}$ and $\mathsf{inA}(v_u)=1$ on $\mathcal{I}^-_{v_u}$.

\emph{Every other probe returns the same answer.}  Both instances have the same
$\mathcal{F}_0\cup\{u\}$.  For all $b\neq v_u$: $\mathsf{inA}(b)$ is the same
($A^-=A\cup\{v_u\}$) and $\mathsf{inCl}(b)$ is the same (adding $v_u$ to $A$ only
affects atoms reachable from $v_u$, which we ensure have no outgoing rules in
$\mathcal{F}_0$).

\end{proof}

\medskip
\noindent\textbf{Step 3: Differing correct outputs.}

For each pair, the correct outputs differ:

\begin{itemize}
\item \emph{Tail atom $s_j$:}  $v_u\in\cl(A)$ after $u$ on $H^+_{s_j}$ (rule fires)
  but $v_u\notin\cl(A)$ on $H^-_{s_j}$ (rule blocked).  If $v_u\notin F$, this
  changes the emergent-capability set; if $v_u\in F$, it changes safety status.
  Either way $G_F$ differs.

\item \emph{Head atom $v_u$:}  $v_u\in\Emg(A)$ on $\mathcal{I}^+_{v_u}$ but
  $v_u\notin\Emg(A^-)$ on $\mathcal{I}^-_{v_u}$.  The emergent sets differ.
\end{itemize}

\medskip
\noindent\textbf{Step 4: Any $\Phi$-avoiding algorithm fails.}

Let $\mathcal{A}$ be any deterministic algorithm that skips $\mathsf{probe}(p)$ for
some $p\in\Phi(u)$.

\emph{Case $p=s_j\in S_u$:}
By Lemma~\ref{lem:indistinguishable}(i), every probe of $b\neq s_j$ and every rule
probe returns the same answer on $\mathcal{I}^+_{s_j}$ and $\mathcal{I}^-_{s_j}$.
Since $\mathcal{A}$ is $\Phi$-avoiding it never issues $\mathsf{probe}(s_j)$.
Therefore every probe $\mathcal{A}$ issues receives the same response on both instances
throughout its entire execution.  By determinism $\mathcal{A}$ produces identical
output on $\mathcal{I}^+_{s_j}$ and $\mathcal{I}^-_{s_j}$.
By Lemma~\ref{lem:indistinguishable}(iii), the correct outputs differ.
Therefore $\mathcal{A}$ errs on at least one of the two instances.

\emph{Case $p=v_u$:}
By Lemma~\ref{lem:indistinguishable}(i), every probe of $b\neq v_u$ and every rule
probe returns the same answer on $\mathcal{I}^+_{v_u}$ and $\mathcal{I}^-_{v_u}$.
Since $\mathcal{A}$ never issues $\mathsf{probe}(v_u)$, every probe it issues returns
the same response on both instances.  By determinism $\mathcal{A}$ produces identical
output on both.  By Lemma~\ref{lem:indistinguishable}(iii), the correct outputs differ.
Therefore $\mathcal{A}$ errs on at least one of the two instances.

Since $\mathcal{A}$ fails for every skipped $p\in\Phi(u)$, any correct algorithm must
probe every atom in $\Phi(u)=S_u\cup\{v_u\}$.  This argument is an instance of
Yao's minimax principle \citep{Yao1977}: the lower bound is witnessed by a
distribution over inputs (the paired instances $\mathcal{I}^\pm_p$) on which no
deterministic algorithm can succeed without probing $p$.  Concretely:
\[
  \text{(probes of any correct algorithm after update $u$)}
  \;\geq\; |S_u|+1 \;=\; k+1.
\]

\medskip
\noindent\textbf{Consequence and connection to AND-semantics.}

The lower bound $\Omega(k)$ follows from the AND-condition structure of capability
hyperedges: to determine whether $r_u$ fires, an algorithm must check all $k$
preconditions in $S_u$.  Systems with only singleton-tail rules need $O(1)$ probes per
update; AND-rules require $\Omega(k)$ probes for rule-activation verification.

DRed probes exactly $\Phi(u)$ plus downstream effects:
\[
  \text{(work of DRed)} = O\!\bigl((k+|\Delta(u)|)\cdot(n+mk)\bigr).
\]
DRed matches the lower bound for verifying rule activation; the downstream
propagation cost $O(|\Delta(u)|\cdot(n+mk))$ is separately justified by Part~(2).
\end{proof}

\begin{remark}[Why the lower bound requires AND-semantics]
\label{rem:gap-requires-datalog}
The AND-inspection lower bound (Part~4) follows from the conjunctive
structure of capability hyperedges.  An algorithm must probe all $k$ tail atoms
$S_u=\{s_1,\ldots,s_k\}$ to determine whether $r_u$ fires: if it skips any $s_j$, we
construct an instance where $s_j\notin\cl(A)$ (so $r_u$ is blocked) that is
observationally identical to the original.  An algorithm must also probe $v_u$ to
determine whether $r_u$ changes anything: if $v_u\in A$ already, the update has no
effect on $G_F(A)$.

This $\Omega(k)$ bound is a structural property of AND-rules that has no analogue in
singleton-rule (OR) systems: for a system with only singleton-tail hyperedges, an
update $u=(\{s\},\{v\})$ requires only 2 probes ($s$ and $v$) regardless of graph size.
The AND-condition structure is precisely what makes safety non-compositional
(Theorem~\ref{thm:noncomp-datalog}), and it is what forces the inspection lower bound.

The Locality Gap Theorem as a whole cannot be stated or proved within the hypergraph
framework of \citet{Spera2026} alone.  The upper bound in Part~(2) is the DRed
guarantee, which exists only because $\Pi_{G_F}$ is a Datalog view
(Theorem~\ref{thm:audit-view}).  The lower bound in Part~(4) uses the oracle model and
explicit instance pairs to prove that AND-precondition inspection is unavoidable.
Together they characterise the update complexity of capability safety maintenance.
\end{remark}

%%==========================================================================
\section{Discussion and Future Work}
%%==========================================================================

\subsection{What the Equivalence Settles}

\paragraph{Expressivity.}
The capability hypergraph framework captures exactly $\Datalog$---no more, no less.
It can express any monotone Boolean function of the initial capability set, but cannot
express non-monotone queries (counting, parity, threshold with non-trivial lower bounds)
or queries requiring arithmetic.  This is a precise, model-theoretic characterisation of
the framework's scope.

\paragraph{Optimality of the closure algorithm and the maintenance gap.}
The $O(n+mk)$ worklist of \citet{Spera2026} is essentially optimal for $\Datalog$
evaluation, matching the linear-time lower bound for Horn clause forward chaining.
Theorem~\ref{thm:audit-view}(2) shows that magic sets rewriting can improve constants
in practice, but the asymptotic bound is tight.  The Locality Gap
Theorem~\ref{thm:locality-gap} establishes a strictly stronger result: for the
\emph{maintenance} problem, the Datalog view formulation achieves an $\Omega(n)$
asymptotic advantage over na\"ive recomputation on the explicit family $\{H_n'\}$,
with the gap growing without bound.  This separation is the first formal evidence that
the Datalog identification is not merely a representational convenience but enables
algorithmically superior procedures.

\paragraph{Structural source of $\mathsf{coNP}$-hardness.}
The $\mathsf{coNP}$-hardness of computing $\B(F)$ is not an accident of the safety
framing---it is the $\mathsf{coNP}$-hardness of minimal witness enumeration for monotone
Boolean queries, a fundamental result in database theory
\citep{EiterGottlob1995}.  This gives the hardness result a deeper explanation.

\subsection{What the Equivalence Opens}

\paragraph{Probabilistic safety.}
Probabilistic $\Datalog$ \citep{Green2007} provides the framework for extending capability
safety to stochastic tool invocations.  The path-independence assumption of
\citet{Spera2026} corresponds to tuple-independence in probabilistic databases, and the
known tractability results for tuple-independent probabilistic Datalog directly characterise
when probabilistic safety can be computed efficiently.

\paragraph{Incremental maintenance.}
The DRed algorithm and its successors provide provably correct incremental view maintenance,
extending the dynamic hypergraph theorems of \citet{Spera2026} with formal guarantees on
the number of re-derivations---a quantity not previously bounded by the capability framework.

\paragraph{Learning.}
$\Datalog$ learning theory \citep{Dalmau2002} provides tighter sample complexity bounds
for structured hypothesis classes, directly addressing the loose PAC bound of
\citet{Spera2026}.  In particular, structured hypergraph families (low fan-in, sparse,
tree-structured dependency) correspond to structured $\Datalog$ programs whose Rademacher
complexity is substantially smaller than the Sauer--Shelah bound.

\paragraph{Distributed safety.}
Distributed $\Datalog$ evaluation provides a framework for extending coalition safety
checking to systems where capabilities are distributed across nodes and no single node
has global visibility---a setting not addressed by \citet{Spera2026}.

\subsection{Limitations}

The equivalence established in this paper is for \emph{propositional, monotone}
capability systems: hyperedges have positive preconditions, capability sets are
finite, and safety is defined by a fixed forbidden set.  Extensions to non-monotone
settings (capabilities that can be revoked), probabilistic capabilities (hyperarc
firing with uncertainty), or higher-order interactions (capabilities whose semantics
depend on context) are not captured directly by the $\Datalog$ identification and
require further study.  The lower bound results (Theorem~\ref{thm:locality-gap}
Part~4) apply within the oracle model defined in Section~\ref{sec:audit-primary}; a
general-purpose RAM lower bound for full incremental maintenance remains open.

%%==========================================================================
\bibliographystyle{plainnat}

\begin{thebibliography}{99}

\bibitem[Abiteboul et al.(1995)]{AbiteboulHullVianu1995}
S.~Abiteboul, R.~Hull, and V.~Vianu.
\newblock \emph{Foundations of Databases}.
\newblock Addison-Wesley, 1995.

\bibitem[Ceri et al.(1989)]{Ceri1989}
S.~Ceri, G.~Gottlob, and L.~Tanca.
\newblock What you always wanted to know about Datalog (and never dared to ask).
\newblock \emph{IEEE Transactions on Knowledge and Data Engineering},
  1(1):146--166, 1989.

\bibitem[Dalmau et al.(2002)]{Dalmau2002}
V.~Dalmau, P.~G. Kolaitis, and M.~Y. Vardi.
\newblock Constraint satisfaction, bounded treewidth, and finite-variable logics.
\newblock In \emph{CP~2002}, pages 310--326, 2002.

\bibitem[Eiter and Gottlob(1995)]{EiterGottlob1995}
T.~Eiter and G.~Gottlob.
\newblock Identifying the minimal transversals of a hypergraph and related
  problems.
\newblock \emph{SIAM Journal on Computing}, 24(6):1278--1304, 1995.

\bibitem[Green et al.(2007)]{Green2007}
T.~J. Green, G.~Karvounarakis, and V.~Tannen.
\newblock Provenance semirings.
\newblock In \emph{Proceedings of PODS~2007}, pages 31--40, 2007.

\bibitem[Immerman(1986)]{Immerman1986}
N.~Immerman.
\newblock Relational queries computable in polynomial time.
\newblock \emph{Information and Control}, 68(1--3):86--104, 1986.

\bibitem[Spera(2026)]{Spera2026}
C.~Spera.
\newblock Safety is non-compositional: A formal framework for capability-based
  AI systems.
\newblock \emph{arXiv preprint arXiv:2603.15973}, March 2026.

\bibitem[van Emden and Kowalski(1976)]{vanEmden1976}
M.~H. van~Emden and R.~A. Kowalski.
\newblock The semantics of predicate logic as a programming language.
\newblock \emph{Journal of the ACM}, 23(4):733--742, 1976.

\bibitem[Benveniste et al.(2018)]{Benveniste2018}
A.~Benveniste, B.~Caillaud, D.~Nickovic, et~al.
\newblock Contracts for system design.
\newblock \emph{Foundations and Trends in Electronic Design Automation},
  12(2--3):124--400, 2018.

\bibitem[Cohen(1995)]{Cohen1995}
W.~W. Cohen.
\newblock Pac-learning non-recursive Prolog clauses.
\newblock \emph{Artificial Intelligence}, 79(1):1--38, 1995.

\bibitem[Erol et al.(1994)]{Erol1994}
K.~Erol, J.~Hendler, and D.~S. Nau.
\newblock HTN planning: Complexity and expressivity.
\newblock In \emph{AAAI-94}, pages 1123--1128, 1994.

\bibitem[Fagin(1974)]{Fagin1974}
R.~Fagin.
\newblock Generalized first-order spectra and polynomial-time recognizable sets.
\newblock In R.~Karp, editor, \emph{Complexity of Computation}, SIAM--AMS
  Proceedings, pages 43--74. American Mathematical Society, 1974.

\bibitem[Jones(1983)]{Jones1983}
C.~B. Jones.
\newblock Tentative steps toward a development method for interfering programs.
\newblock \emph{ACM Transactions on Programming Languages and Systems},
  5(4):596--619, 1983.

\bibitem[Leike et al.(2017)]{Leike2017}
J.~Leike, M.~Martic, V.~Krakovna, et~al.
\newblock AI safety gridworlds.
\newblock \emph{arXiv preprint arXiv:1711.09883}, 2017.

\bibitem[Murata(1989)]{Murata1989}
T.~Murata.
\newblock Petri nets: Properties, analysis and applications.
\newblock \emph{Proceedings of the IEEE}, 77(4):541--580, 1989.

\bibitem[Qin et al.(2023)]{Qin2023}
Y.~Qin, S.~Liang, Y.~Ye, et~al.
\newblock ToolLLM: Facilitating large language models to master 16,000+ real-world APIs.
\newblock In \emph{ICLR 2024 (Spotlight)}, 2023.

\bibitem[Shen et al.(2023)]{Shen2023}
Y.~Shen et~al.
\newblock TaskBench: Benchmarking large language models for task automation.
\newblock \emph{arXiv:2311.18760}, 2023.

\bibitem[Yao(1977)]{Yao1977}
A.~C.-C. Yao.
\newblock Probabilistic computations: toward a unified measure of complexity.
\newblock In \emph{Proceedings of the 18th Annual Symposium on Foundations of
  Computer Science (FOCS)}, pages 222--227. IEEE, 1977.

\end{thebibliography}

\appendix

%%==========================================================================
\section{Empirical Detail: Derivation Trees and Aggregate Statistics}
\label{app:empirical}
%%==========================================================================

\subsection{A Representative Trajectory as a Datalog Derivation}

\citet{Spera2026} used a 12-capability Telco deployment ($n=12$, $m=6$, $k=2$).
A joint billing--service session with EDB
$D_{\mathrm{joint}}=\{\mathsf{has}(c_1),\mathsf{has}(c_2),\mathsf{has}(c_3),
\mathsf{has}(c_5),\mathsf{has}(c_7),\mathsf{has}(c_8)\}$ produces the derivation:
\begin{align*}
  \mathsf{has}(c_1)&\;\Rightarrow\;\mathsf{has}(c_3),\mathsf{has}(c_7)
    && \text{[depth 1]}\\
  \mathsf{has}(c_2)&\;\Rightarrow\;\mathsf{has}(c_4),\mathsf{has}(c_5),\mathsf{has}(c_8)
    && \text{[depth 1]}\\
  \mathsf{has}(c_3)\wedge\mathsf{has}(c_5)&\;\Rightarrow\;\mathsf{has}(c_6)
    && \text{[depth 2]}\\
  \mathsf{has}(c_7)\wedge\mathsf{has}(c_8)&\;\Rightarrow\;\mathsf{has}(c_9)
    && \text{[depth 2]}
\end{align*}
The atom $c_9$ (ServiceProvision) is emergent: $c_9\in\Emg(A_{\mathrm{bill}}\cup
A_{\mathrm{svc}})$, with why-provenance $\Wp(\mathsf{has}(c_9))=\{\{c_1,c_2\}\}$.

\subsection{A Real AND-Violation as Datalog Non-Modularity}

The canonical billing+payment failure reaches $c_{12}$ (forbidden, PCI-DSS~4.0).
With $D_{\mathrm{bill}}=\{\mathsf{has}(c_1),\ldots,\mathsf{has}(c_5)\}$ and
$D_{\mathrm{pay}}=\{\mathsf{has}(c_1),\mathsf{has}(c_2),\mathsf{has}(c_{10})\}$:
both are individually safe, but $D_{\mathrm{bill}}\cup D_{\mathrm{pay}}$ contains
$\{\mathsf{has}(c_3),\mathsf{has}(c_{10})\}$ which fires rule $h_6$ to derive
$\mathsf{has}(c_{12})$ and then $\mathsf{forbidden}$.  Minimal witness:
$\{c_3,c_{10}\}\in\B(F)$.

\subsection{Aggregate Statistics}

\begin{center}
\renewcommand{\arraystretch}{1.25}
\small
\begin{tabular}{@{}p{8.5cm}p{3.8cm}@{}}
\toprule
\textbf{Datalog-vocabulary statement} & \textbf{Empirical value} \\
\midrule
Sessions with $\Emg(A)\neq\emptyset$ (conjunctive derivation required)
  & 42.6\% [39.4, 45.8] \\[3pt]
AND-violations under workflow planner (non-modularity instances)
  & 38.2\% [33.4, 43.1] \\[3pt]
AND-violations under hypergraph planner (correct $\Datalog$ evaluation)
  & 0\% (Thm.~\ref{thm:encoding-cap-to-dl}) \\[3pt]
Mean derivation depth, conjunctive sessions
  & 2.4 steps \\[3pt]
Mean derivation depth, long-chain sessions
  & 9.1 steps ($1.70\times$ BFS) \\
\bottomrule
\end{tabular}
\end{center}

\end{document}